\documentclass[11pt]{amsart}

\usepackage{stmaryrd}
\usepackage{amscd,amsxtra,amssymb,mathrsfs, bbm, bigints}

\usepackage{multirow}
\usepackage[pdftex]{hyperref}
\usepackage{longtable}
\usepackage{textgreek}
\usepackage{tikz}
\usepackage[utf8]{inputenc}
\usepackage{float}
\usepackage{color}
\usetikzlibrary{calc, positioning, shapes, fit, matrix, decorations}
\usetikzlibrary{decorations.shapes}

\usepackage[cmtip]{xypic}
\usepackage[all]{xy}

\xyoption{arc}

\newtheorem{theorem}{Theorem}[section]

\newtheorem{lemma}[theorem]{Lemma}
\newtheorem{proposition}[theorem]{Proposition}

\theoremstyle{definition}
\newtheorem{definition}[theorem]{Definition}
\newtheorem{remark}[theorem]{Remark}
\newtheorem{example}[theorem]{Example}

\numberwithin{equation}{section}

\DeclareMathAlphabet{\mathpzc}{OT1}{pzc}{m}{it}

\renewcommand{\dim}{\mathsf{dim}}

\DeclareMathOperator{\Hom}{\mathsf{Hom}}

\DeclareMathOperator{\Mat}{\mathsf{Mat}}

\input xy
\xyoption{all}

\newcommand{\bz}{{\bar z}}

\newcommand{\kA}{\mathcal{A}}
\newcommand{\kB}{\mathcal{B}}
\newcommand{\kE}{\mathcal{E}}

\newcommand{\kO}{\mathcal{O}}
\newcommand{\kL}{\mathcal{L}}

\newcommand{\kW}{\mathcal{W}}

\newcommand{\lar}{\longrightarrow}

\newcommand{\RR}{\mathbb R}
\newcommand{\CC}{\mathbb C}

\newcommand{\ZZ}{\mathbb{Z}}

\newcommand{\nn}{\mathsf{n}}

\def\Mat{\mathop\mathrm{Mat}}

\def\8{\infty}			
	\def\+{\oplus}		
\def\*{\otimes}

\def\be{\begin{equation}}
\def\ee{\end{equation}}

\def\kA{\mathcal A} 
\def\kB{\mathcal B} \def\kO{\mathcal O}

\def\kE{\mathcal E}

 \def\kW{\mathcal W}
 
\def\kL{\mathcal L}

	\def\NN{\mathbb N}

\newcommand{\rightarrowdbl}{\longrightarrow\mathrel{\mkern-14mu}\rightarrow}

\def\DMO{\DeclareMathOperator}
\DMO{\ob}{Ob}            \DMO{\mor}{Mor}
\DMO{\Ker}{Ker}
\DMO{\id}{Id}

\title[Norms of wave functions for FQHE models on a torus]{Norms of wave functions for FQHE models on a torus}

\author{Igor Burban}
\address{
Universit\"at Paderborn,
Institut f\"ur Mathematik,
Warburger Strasse 100,
33098 Paderborn,
Germany
}
\email{burban@math.uni-paderborn.de}

\author{Semyon Klevtsov}
\address{
IRMA, Universit\'e de Strasbourg,
UMR 7501, 7 rue Ren\'e  Descartes,
67084 Strasbourg, 
France
}
\email{klevtsov@unistra.fr}

\begin{document}

\begin{abstract}
The goal of this paper is to give an explicit computation of the curvature of the magnetic vector bundle of the multi-layer model of the fractional quantum Hall effect on a torus. We also obtain concrete formulae for the norms of the corresponding wave functions arising in such models.
\end{abstract}

\maketitle

\section{Introduction} 
The quantum-mechanical behaviour of a non-relativistic  charged particle moving in a plane   in presence of a perpendicular uniform magnetic field  is described by the magnetic Schr\"odinger operator 
\begin{equation}\label{E:LandauPlane}
H : = - \left(\frac{\partial}{\partial u} + i \frac{b}{2} v \right)^2 - \left(\frac{\partial}{\partial v} - i \frac{b}{2} u\right)^2: \; D(H) \lar L_2(\RR^2)
\end{equation}
where $(u, v)$ are the standard coordinates on $\RR^2$ and $b \in \RR_{> 0}$ is a parameter  describing  the strength of the magnetic field. The operator $H$ is essentially self-adjoint and its definition domain $D(H)$ is the entire Hilbert space $L_2(\RR^2)$.  Moreover, the spectrum of $H$ is the set 
$\left\{(2l+1) \, \big| \,  l \in \NN_0 \right\}$ and its ground state is the Segal--Bargmann--Fock space
\begin{equation}
\left\{ \left. e^{-\frac{b}{4}|z|^2} f(z) \,\right| \, \CC \stackrel{f}\lar \CC \; \mbox{is analytic}  \right\} \cap L_2(\RR^2),
\end{equation}
where $z = u + i v$; see for instance  \cite{LandauLifshitz} for details.

The  magnetic Schr\"odinger operator, describing the behaviour of $n$ charged particles moving in $\RR^2$,  is given by the expression 
\begin{equation}\label{E:ClassicMagneticSchroedinger}
H_{n, w} : = - \sum\limits_{k = 1}^n \left( \left(\frac{\partial}{\partial u_k} + i \frac{b}{2} v_k\right)^2 + \left(\frac{\partial}{\partial v_k} - i \frac{b}{2} u_k\right)^2\right) + \sum\limits_{k < l} w\bigl(|z_k - z_l|),
\end{equation}
where $z_k = (u_k, v_k)$  for $1 \le k \le n$ and $w$ is some potential function, which describes the mutual interactions between the particles. 
For $w \ne 0$, explicit expressions for  wave functions of $H_{n, w}$ are rarely known. 

The fractional quantum Hall effect (FQHE) is a celebrated physical phenomenon in which the Hall conductance of a two-dimensional  electron system  in a perpendicular magnetic field shows precisely quantized plateaus at fractional values of  $e^{2}/h$, where $e$ is the electron charge and $h$ is the Planck constant. In order to explain this phenomenon, Laughlin \cite{Laughlin} suggested to replace  the ground state of $H_{n, w}$
by the function
\begin{equation}\label{E:LaughlinAnsatz}
\Phi(z_1, \dots, z_n) := e^{-\frac{b}{4}\left(|z_1|^2 + \dots + |z_n|^2\right)} \biggl(\prod\limits_{k < l} (z_k - z_l)\biggr)^{m},
\end{equation}
where $m \in \NN_0$ is a positive integer. Note that $\Phi$ is anti-symmetric if $m $ is odd (realizing the fermionic statistics in the case of electrons) and symmetric if $m$ is even (corresponding to the bosonic statistics)  and belongs to the ground state of the operator $H_{n, 0}$. 

Investigations of many-body systems, which include effects of the boundary conditions, are very hard to perform. A convenient way to circumvent such difficulties  is provided by replacing the plane $\CC = \RR^2$ by a compact Riemann surface \cite{H,WN,ASZ,KMMW}. A complex torus $E = \CC/\Lambda$ is a particularly convenient choice \cite{HaldaneRezayi}, see \cite{Avron} and also \cite{K16} for a survey.
 Here, $\Lambda = \langle 1, \tau\rangle_{\mathbbm{Z}} \subset \CC$ is a  lattice determined by a modular parameter  $\tau \in \CC$ with $t = \mathfrak{Im}(\tau) > 0$. 
We equip $E$ with a Riemannian metric induced by the quotient $\RR^2 = \CC  \rightarrowdbl E$.  Let $\kL$ be a holomorphic line bundle of degree $k$.  The moduli space  of such line bundles can be identified with the torus $E$ itself, hence we use the notation $\kL = \kL_{k, \xi}$ with $\xi = a \tau + b$ for $a, b \in \RR$. There is a convenient way to define $\kL_{k, \xi}$ using the so-called automorphy factors. It turns out that  $\kL_{k, \xi}$ can be equipped with a natural  hermitian metric $h = h_{k, \xi}$, whose curvature form is a constant two-form on $E$ corresponding to the uniform magnetic field. 

The space $\mathsf{B}_{k, \xi}$ of smooth global sections of $\kL_{k, \xi}$ carries a natural hermitian inner product given for any $f, g \in \mathsf{B}_{k,\xi}$ by the rule
$
 \langle f, g \rangle := \iint\limits_{E} h(f, g) \, \omega,
$
where  $h(f, g): E \lar \CC$ is a smooth function assigning to a point $p \in E$ the scalar product of sections $f(p), g(p) \in \kL\big|_p$ with respect to the metric $h$ and $\omega$ is a volume form on $E$. One  can define  the corresponding Hilbert space $\mathsf{H}_{k, \xi}$ of square integrable global section of $\kL_{k, \xi}$ in a standard way.

 The quantum-mechanical behaviour of a charged particle moving on $E$ in presence of a uniform magnetic field is described by the so-called holomorphic Bochner Laplacian $D(H) \stackrel{H}\lar 
\mathsf{H}_{k, \xi}$ of  the hermitian line bundle $\kL_{k, \xi}$. The operator $H$ is essentially self-adjoint and has a discrete spectrum with finite dimensional eigenspaces,  which can be realized as subspaces of $\mathsf{B}_{k, \xi}$. As in the case of a plane, we have $D(H) = \mathsf{H}_{k, \xi}$.  Under the assumption that  $\kL$ has positive degree, the ground state of $H$ coincides with the space $\mathsf{W}_{k,\xi}$ of \emph{holomorphic} global sections of $\kL$; see for instance \cite{MaMarinescu, Prieto} for details and proofs.

For any $k > 0$, the space $\mathsf{W}_{k, \xi}$  is $k$-dimensional and admits a distinguished basis $(s_1, \dots, s_k)$  where
\begin{equation*}\label{E:BasisThetaFunctionsIntro}
s_j(z) = \vartheta\Bigl[\frac{j-1}{k}, 0\Bigr](kz + \xi | k \tau) \; \mbox{\rm for} \; 1 \le j \le k,
\end{equation*}
where $\vartheta[a, b](z \,\tau)$ is the theta-function with characteristics $a, b \in \RR$ (see \cite{Mumford}).
Following the Ansatz of Laughlin (\ref{E:LaughlinAnsatz}), Haldane any Rezayi suggested in \cite{HaldaneRezayi} to consider the $m$-dimensional vector space $\mathsf{V}_{m, n, \xi}$ spanned by the functions 
\begin{equation*}\label{E:HRSpaceIntro}
\Phi_j(z_1, \dots, z_n) := \vartheta\left[\dfrac{j-1}{m}, 0\right](mw + \xi\,|\,m\tau) \left(\prod\limits_{p < q} \vartheta(z_p - z_q)\right)^m
\end{equation*}
for $1 \le j \le m$. Here,   $w = z_1 + \dots + z_n$ and $\vartheta(z) = \vartheta\bigl[\frac{1}{2}, \frac{1}{2}\bigr](z|\tau)$ (an important feature is that  $\vartheta(z)$ is an odd function). Note that all functions $\Phi_j$ are anti-symmetric in the case $m$ is odd and symmetric if $m$ is even. 

In our previous work \cite{BurbanK}, we gave an axiomatic characterization  of the space $\mathsf{V}_{m, n, \xi}$. For simplicity assume  that $m$ is even. Then $\mathsf{V}_{m, n, \xi}$ is precisely the space of global sections of the line bundle $\underbrace{\kL_{mn, \xi} \boxtimes \dots \boxtimes \kL_{mn, \xi}}_{n  \; \mbox{\scriptsize{\sl times}}}$ on $X := \underbrace{E \times \dots \times E}_{n  \; \mbox{\scriptsize{\sl times}}}$, which vanish with order at least $m$ on each partial diagonal of $X$. In particular,  $\mathsf{V}_{m, n, \xi}$ is a subspace of the ground state $\underbrace{\mathsf{W}_{mn, \xi} \otimes \dots \otimes 
\mathsf{W}_{mn, \xi}}_{n  \; \mbox{\scriptsize{\sl times}}}$ of an appropriate $n$-body magnetic Schr\"odinger operator and inherits the corresponding  hermitian inner product.  

Based on this fact,  we constructed in \cite{BurbanK} the so-called magnetic vector bundle $\kB$ on the torus $E$. It is a holomorphic hermitian vector bundle of rank $m$ such that for any $\xi \in \CC$ the fiber $\kB\big|_{[\xi]}$ is isomorphic (as a hermitian complex vector space) to  $\mathsf{V}_{m, n, \xi}$. Using the technique of Fourier--Mukai transforms, we showed that $\kB$ is simple and has degree $-1$. 

More generally, one can define multi-layer  torus models  of FQHE \cite{Halperin,KeskiVakkuriWen, Wen} starting with a so-called Wen datum $(K, \vec{\mathsf{n}})$. Here,  $K \in \mathsf{Mat}_{g \times g}(\NN_0)$ is a symmetric positive definite matrix, whose diagonal elements are simultaneously  all even or odd,  and $\vec{\mathsf{n}} = \left(\begin{array}{c} 
n_1 \\ \vdots \\ n_g\end{array}\right) \in \NN^g$ is such that 
 $K \vec{\mathsf{n}} = d 
\left(\begin{array}{c} 
1 \\ \vdots \\ 1 \end{array}\right)
$
for some $d \in \NN$. Then for any $\vec\zeta \in \CC^g$, one can consider  a natural  vector space $\mathsf{V}_{K, \vec{\mathsf{n}}, \vec\zeta}$ of wave functions of Keski-Vakkuri and Wen \cite{KeskiVakkuriWen, Wen} having  dimension $\delta = \det(K)$ (coinciding with vector the space $\mathsf{V}_{m, n, \zeta}$ if $g = 1$, $K = (m)$ and $\vec{\mathsf{n}} = n$). In \cite{BurbanK}, we constructed a holomorphic hermitian vector bundle on the abelian variety $\CC^g/(\ZZ^g + \tau \ZZ^g)$ such that such that for any $\vec\zeta \in \CC$ its fiber over $[\vec\zeta] \in \CC^g/(\ZZ^g + \tau \ZZ^g)$  is isometric  to the hermitian complex vector space $\mathsf{V}_{K, \vec{\mathsf{n}}, \vec\zeta}$. Considering only 
$\vec\zeta = (\xi, \dots, \xi)$ with  $\xi \in E$, we get a hermitian holomorphic vector bundle $\kB$ on the torus $E$.  

Any holomorphic hermitian vector bundle on a complex manifold carries a canonical hermitian connection $\nabla$ (the Bott--Chern connection \cite{BottChern}). In \cite{BurbanK} is was shown that the Bott--Chern connection of $\kB$ (under appropriate additional assumptions on $(K, \vec{\mathsf{n}})$) is projectively flat. 
Elaborating this result one step further, we compute  the curvature form of $\nabla$, which can be identified with the  Hall conductance of the corresponding many-particle system on $E$, as was first explained in a geometric context by Avron-Seiler-Zograf \cite{ASZ}, see also \cite{KMMW}. 
It turns out, the degree of the magnetic vector bundle $\kB$ is $-\dfrac{n\delta}{d}$, where $n = n_1 + \dots + n_g$ is the total  number particles of the system.  Hence, the absolute value of  the slope of $\kB$ is $\dfrac{n}{d}$, see Theorem \ref{T:CurvatureLaughlinBundle}. In \cite{BurbanK} is was proven that $\kB$ is a simple vector bundle. 

 Whereas our previous work \cite{BurbanK} was essentially based on  the theory of derived categories and Fourier--Mukai transforms, the approach of this paper is more elementary. The core  results of this work are explicit formulae for  norms of many-body wave functions of FQHE models on a torus (see Lemma \ref{L:NormManyBody},  Proposition \ref{P:Slater}, Lemma \ref{L:NormMultiLayer} and Proposition \ref{P:NormCenterMass}), based on direct computations with theta-functions.

\section{Elliptic curve: one-particle case}

In this work, $E := \CC/\Lambda$ is a complex torus,  $\Lambda = \langle 1, \tau\rangle_{\mathbbm{Z}} \subset \CC$ is a  lattice determined by the modular parameter $\tau \in \CC$ with $t = \mathfrak{Im}(\tau) > 0$.
Recall (see \cite[Chapter I]{Mumford}) that  the theta-function with characteristics $a, b \in \RR$ is defined as:
\begin{equation}\label{E:ThetaWithCharacteristics}
\vartheta[a, b](z | \tau):= \sum\limits_{n \in \ZZ} 
\exp\bigl(\pi i \tau (n+a)^2+2\pi i (n+a)(z+b)\bigr).
\end{equation}
As a function of $z$, the theta-series satisfies the following quasi-periodic conditions:
\begin{equation}\label{E:ThetaFunctionsTransfRules}
\left\{
\begin{array}{lcl}
\vartheta[a, b](z + 1 | \tau ) & = & e^{2\pi i a} \vartheta[a, b](z | \tau) \\
\vartheta[a, b](z + \tau | \tau) & = & e^{- 2\pi i (z +b) - \pi i \tau} \vartheta[a, b](z | \tau).
\end{array}
\right.
\end{equation}
Considering $\tau$ as a constant, we put: 
\begin{equation}\label{E:ThetaOdd}
\vartheta(z) :=  \vartheta\left[\frac{1}{2}, \frac{1}{2}\right](z | \tau).
\end{equation}
Then  $\vartheta(-z) = - \vartheta(z)$ and $\vartheta$ has a unique simple zero at $z = 0$ modulo the lattice $\Lambda$. 

For any holomorphic function $\CC \stackrel{\psi}\lar \CC^\ast$ satisfying $\psi(z+1) = \psi(z)$ and $c \in \CC^\ast$,  we define the holomorphic line bundle
$
\kL(c, \psi)$ on $E$ via the commutative diagram of holomorphic complex manifolds
$$
\xymatrix
{
\CC \times V \ar[r] \ar[d]_{\mathrm{pr}_1} & \kL(c, \psi) := \CC \times V /\sim
\ar[d]\\
\CC \ar[r]^\pi  & E
}
$$
Here,  $V = \CC$ and $(z, v) \sim (z+1, cv) \sim (z + \tau, \psi(z) v)$.

Let $\varphi (z) := \exp(-\pi i \tau - 2 \pi i z)$.
Then any  holomorphic line bundle on $E$ of degree $k \in \ZZ$ is isomorphic to  $\kL_{k, \xi} := 
\kL\bigl(1, e^{-2\pi i \xi}\varphi(z)^k\bigr)$, with $\xi = a \tau + b$  and $a, b \in \RR/\ZZ$. Moreover, $\kL_{k, \xi} \cong \kL_{k, \zeta}$ if and only if $[\xi] = [\zeta]$ in $\CC/\Lambda = E$. 
Next, the  space $\mathsf{B}_{k, \xi}$ of smooth global sections of  $\kL_{k, \xi}$ admits the following description: 
\begin{equation}\label{E:SmoothSections}
\mathsf{B}_{k, \xi} := \left\{
\CC \stackrel{f}\lar \CC \left| \, \begin{array}{l}
f \; \mbox{\rm is smooth}, f(z+1) = f(z) \\
f(z + \tau) = e^{-2\pi i \xi} \varphi(z)^k f(z)
\end{array}
\right.
\right\}
\end{equation}
For $k > 0$, the space $\mathsf{W}_{k, \xi}:= \Gamma\bigl(E, \kL_{k, \xi})$ of holomorphic sections of the line bundle $\kL_{k, \xi}$ has dimension $k$ and has a distinguished basis $(s_1, \dots, s_k)$, where
\begin{equation}\label{E:BasisThetaFunctions}
s_j(z) = \vartheta\Bigl[\frac{j-1}{k}, 0\Bigr](kz + \xi | k \tau) \; \mbox{\rm for} \; 1 \le j \le k.
\end{equation}
We set $z = x + \tau y$, where $x, y\in \RR$ and consider the following function
\begin{equation}\label{E:MetricLB1}
h(z) = h_{k, \xi}(x, y) := \exp(- 2 \pi kty^2  - 4 \pi a ty) = \exp \bigl(2\pi t \frac{a^2}{k}\bigr) \exp\left(-2\pi kt\left(y + \frac{a}{k}\right)^2\right).
\end{equation}
This function defines a hermitian metric on the holomorphic line bundle $ \kL_{k, \xi}$, so that
for any $f \in \mathsf{B}_{k, \xi}$, we have a smooth function
\begin{equation}\label{E:NormSection}
E \lar {\RR_{\geqslant0}}, \; [z]  \mapsto  |f(z)|^2 \,h(z)
\end{equation}
assigning to each point $p = [z] \in E$ the square of the length of the value of the section $f$ at $p$. 
This follows from the transformation laws $$ h(x+1, y) = h(x,y) \;\;\; \mbox{and} \;\;\; h(x, y+1) = e^{- 2 \pi t(2 ky + 2 a + k)} h(x, y)$$ 
combined with the quasi-periodicity properties of $f$ given by (\ref{E:SmoothSections}).

The above choice of a hermitian metric $h$  combined with a choice of a volume form  $\omega = dx\wedge dy=\frac{i}{2t} dz\wedge d\bz$ on $E$ allows to define  the Hilbert space of square integrable global sections of the hermitian line bundle $\kL_{k, \xi}$: 
\begin{equation}\label{E:Hilbert space}
\mathsf{H}_{k, \xi} := \left\{
\CC \stackrel{f}\lar \CC \left| \, \begin{array}{l}
 f(z+1) = f(z), f(z + \tau) = \exp(-2\pi i \xi) \varphi(z)^k f(z) \\
\iint\limits_{[0, 1]^2} e^{-2\pi k ty^2 - 4 \pi aty}
\big|f(x, y)\big|^2 \omega < \infty
\end{array}
\right.
\right\}.
\end{equation}
In particular, for any $f, g \in \mathsf{H}_{k, \xi}$, their scalar product is given by the formula
\begin{equation}\label{E:HermitianProduct}
\langle f, g\rangle =  \iint\limits_{[0, 1]^2}   e^{-2\pi k ty^2 - 4 \pi aty}
f(z) \overline{g(z)} dx \wedge dy.
\end{equation}
Let $z = u + i v$ with $u, v \in \RR$ and $\bar\partial = \frac{1}{2}(\partial_u + i \partial_v)$.  Then we have a $\CC$-linear map $\mathsf{B}_{k, \xi} \stackrel{\bar\partial}\lar 
\mathsf{B}_{k, \xi}$, which admits an adjoint operator given by the formula
$\bar\partial^\ast = - \left(\partial + \dfrac{\partial(h)}{h}\right)$. It can be shown that  
\begin{equation}\label{E:Torus}
H^{tor} := \bar\partial^\ast \bar\partial 
= -\dfrac{1}{4}\bigl(\partial_u^2 + \partial_v^2 \bigr) - 2\pi i \dfrac{kv}{t}\bar\partial.
\end{equation}
is 
an essentially self-adjoint operator
acting 
on the Hilbert space $\mathsf{H}_{k, \xi}$, which can be identified with   the  {holomorphic Bochner Laplacian} of $\bigl(\kL_{k, \xi}, h\bigr)$. The magnetic Schr\"odinger operator  $H^{tor}$  given  by (\ref{E:Torus}) is an analogue of the plane operator (\ref{E:LandauPlane}): it describes quantum mechanical motion of a charged particle on a torus $E$ in presence  of a constant magnetic field, see \cite{DN,On,Prieto,DK} and references therein.

It follows that for $k > 0$, the ground state of the operator  $H^{tor}$ is the vector space 
$\mathsf{W}_{k, \xi}$ and the functions $s_1, \dots, s_k$ form its basis.

\begin{proposition}\label{P:ScalarProductTheta}
For any $1 \le p, q \le k$ we have:
\begin{equation}
\langle s_p, s_q\rangle = \delta_{pq} \sqrt{\frac{1}{2kt}} 
e^{2\pi t \frac{a^2}{k}}.
\end{equation}
\end{proposition}
\begin{proof}
Since the theta-series \eqref{E:ThetaWithCharacteristics} converge absolutely and uniformly in $z$ on compact sets of $\mathbb C$, the integration and summation can be interchanged in the series expansions for the scalar product: 
\begin{align}\nonumber
\langle s_p, s_q\rangle &=\iint\limits_{[0, 1]^2} e^{-2\pi k ty^2 - 4 \pi aty}
\vartheta\Bigl[\frac{p-1}{k}, 0\Bigr](kz + \xi | k \tau)\overline{\vartheta\Bigl[\frac{q-1}{k}, 0\Bigr](kz + \xi | k \tau)}
 dxdy\\\nonumber
&=\sum_{m,n\in\mathbb Z} e^{{\pi ik\tau\left(m+\frac{p-1}k\right)^2-\pi ik\bar\tau\left(n+\frac{q-1}k\right)^2}} S_{m, n},\\ \nonumber
\end{align}
where 
\begin{equation}\label{E:Summands}
S_{m, n} = \iint\limits_{[0, 1]^2} e^{-2\pi k ty^2 - 4 \pi aty} F_{n, m}(x, y) dx dy
\end{equation}
with 
\begin{align}\nonumber
F_{m, n}(x, y)   = e^{2\pi i \bigl((k(n-m)+(p-q)\bigr) x}   \cdot
e^{2\pi i \Bigl(\bigl(n + \frac{p-1}{k}\bigr)\tau - \bigl(m + \frac{q-1}{k}\bigr)\bar{\tau}\Bigr) ky}\cdot  \\ \nonumber
 e^{2\pi i \Bigl(\bigl(n + \frac{p-1}{k}\bigr)(b + a\tau) - \bigl(m + \frac{q-1}{k}\bigr)(b + a \bar{\tau})\Bigr)}. \nonumber
\end{align}
Since $|p-q| < k$, we have: 
$$
\int\limits_{0}^1 e^{2\pi i \bigl(k(n-m)+(p-q)\bigr) x} dx = 
\left\{
\begin{array}{cl}
1
& \mbox{\rm if} \;  n = m \;  \mbox{and} \;   p = q  \\
0
 & \mbox{\rm otherwise}.
\end{array}
\right. 
$$
Suppose first that  $p \ne q$. 
We can separate the variables  $x$ and $y$  in the integral (\ref{E:Summands}) to  conclude that  $S_{m, n} = 0$ for all $m, n \in \ZZ$. In particular, $\langle s_p, s_q\rangle = 0$ for $p \ne q$.

\smallskip
\noindent
From now on, assume that $p = q$. Since $S_{m, n} = 0$ for all $m \ne n$,  we have: 
\begin{align}\nonumber
\langle s_p,  s_p\rangle &=\sum_{m\in\mathbb Z}e^{-2\pi kt\left(m+\frac{p-1}k\right)^2}\int\limits_{[0, 1]} e^{-2\pi k ty^2 - 4 \pi aty-4\pi t\left(m+\frac{p-1}k\right)(ky+a)}dy\\ \nonumber
&=\sum_{m\in\mathbb Z}\;\;\int\limits_{[0, 1]} e^{-2\pi kt\left(y+m+\frac{p-1+a}k\right)^2+2\pi t\frac{a^2}k}dy=\sum_{m\in\mathbb Z}\;\;\int\limits_{[m, m+1]} e^{-2\pi kt\left(y+\frac{p-1+a}k\right)^2+2\pi t\frac{a^2}k}dy\\&=\int_{-\infty}^{\infty}e^{-2\pi kty^2+2\pi t\frac{a^2}k}dy=\frac1{\sqrt{2kt}}\cdot e^{2\pi t\frac{a^2}k}.\nonumber
\end{align}
Proposition is proven.
\end{proof}
\begin{remark} The fact that $\langle s_p, s_q\rangle = 0$  and $\langle s_p, s_p\rangle = \langle s_q, s_q\rangle$ for all for $1 \le p \ne q \le k$
 can be also shown using an appropriate unitary action of a finite Heisenberg group on the space $\mathsf{W}_{k, \xi}$; see  \cite[Section 2]{BurbanK}.
\end{remark}

\section{Laughlin states on a complex torus}
In this section, we compute scalar products of many-body wave functions on a torus. Let $n \in \NN$ be the number of particles, $m \in \NN$ another parameter (whose meaning will become clear later) and $\xi = a\tau + b$ with $a, b \in \RR$.  We need the following notation. 

Let $X := \underbrace{E \times \dots \times E}_{n  \; \mbox{\scriptsize{\sl times}}}$ and 
$
 \epsilon_{m, n} := (-1)^{m(n-1)} = 
 \left\{
\begin{array}{cl}
\, 1
& \mbox{\rm if} \, m \in 2 \NN \; \mbox{\rm or} \; m, n \in 2 \NN_0+1  \\
-1
 & \mbox{\rm if} \; m \in 2 \NN_0+1 \;  \mbox{\rm and} \;  n \in 2\NN.
\end{array}
\right. 
 $
 Next, we put $\kL_{k, \xi}^\sharp :=  \kL(-1, - e^{-2\pi i \xi}\varphi(z)^k\bigr)$. It can be shown that $\kL_{k, \xi}^\sharp \cong \kL_{k, \xi^\sharp}$, where  $\xi^\sharp = \xi + \dfrac{1+\tau}{2}$. Moreover, the function $h_{k, \xi}$ given by (\ref{E:MetricLB1}), defines a hermitian metric on  the line bundle $\kL_{k, \xi}^\sharp$ by the same formula (\ref{E:NormSection}). For $k > 0$, we denote by $\mathsf{W}_{k, \xi}^\sharp$ the space of global holomorphic sections of $\kL_{k, \xi}^\sharp$:
\begin{equation*}
\mathsf{W}_{k, \xi}^\sharp := \left\{
\CC \stackrel{f}\lar \CC \left| \, \begin{array}{l}
f \; \mbox{\rm is holomorphic}, f(z+1) = -f(z) \\
f(z + \tau) = -e^{-2\pi i \xi} \varphi(z)^k f(z)
\end{array}
\right.
\right\}.
\end{equation*} 
 Next, we define the following line bundle on $X$: 
\begin{equation*}
\kW_{m, n, \xi} := \left\{
\begin{array}{cl}
\kL_{mn,_\xi} \boxtimes \dots \boxtimes \kL_{mn, \xi}
& \mbox{\rm if} \; \epsilon_{m, n} = \; 1  \\
\kL_{mn,_\xi}^\sharp \boxtimes \dots \boxtimes \kL^\sharp_{mn, \xi}
 & \mbox{\rm if} \; \epsilon_{m, n} = -1.
\end{array}
\right. 
\end{equation*}
By K\"unneth formula, we have:
$$
\mathsf{W}_{m, n, \xi}:= \Gamma(X, \kW_{m, n, \xi}) = 
\left\{
\begin{array}{cl}
\mathsf{W}_{mn, \xi} \otimes \dots \otimes \mathsf{W}_{mn, \xi}
& \mbox{\rm if} \; \epsilon_{m, n} = \; 1  \\
\mathsf{W}_{mn, \xi}^\sharp \otimes \dots \otimes \mathsf{W}_{mn, \xi}^\sharp
 & \mbox{\rm if} \; \epsilon_{m, n} = -1.
\end{array}
\right.
$$
In other terms, $\mathsf{W}_{m, n, \xi}$ can be identified with the vector space
\begin{equation*}
 \left\{
\CC^n \stackrel{\Phi}\lar \CC \left| \, \begin{array}{l}
\Phi \; \mbox{\rm is holomorphic} \\
 \Phi(z_1, \dots, z_k+1, \dots, z_n) = 
\epsilon_{m,n} \Phi(z_1, \dots, z_k+1, \dots, z_n) \\
\Phi(z_1, \dots, z_k+\tau, \dots, z_n) = 
\epsilon_{m,n} \varphi(z_k)^{mn}\Phi(z_1, \dots, z_k+1, \dots, z_n)
\end{array}
\right.
\right\}
\end{equation*} 
for all  $1 \le k \le n$. 

The made choice (\ref{E:MetricLB1}) of a  hermitian metric on the line bundle $\kL_{mn, \xi}$ (respectively, 
$\kL_{mn, \xi}^\sharp$) induces a hermitian metric on the line bundle $\kW_{m, n, \xi}$, given (regardless of the parity of $\epsilon_{m,n}$) by the function
\begin{equation}\label{E:MetricManyBody}
\widehat{h}(z_1, \dots, z_n) = h_{mn, \xi}(z_1) \dots h_{mn, \xi}(z_n) = 
\exp\Bigl(- 2 \pi mn t \sum\limits_{k = 1}^n y_k^2  - 4 \pi a t \sum\limits_{k = 1}^n y_k\Bigr),
\end{equation}
where $z_k = x_k + \tau y_k$ with $x_k, y_k \in \RR$. We choose the volume form $dx_1 \wedge dy_1 \wedge \dots \wedge dx_n \wedge dy_n$ on the abelian variety $X$. Then the vector space $\mathsf{W}_{m, n, \xi}$ gets equipped with the following hermitian inner product
\begin{equation}\label{E:ScalarP}
\langle \Phi, \Psi\rangle := \iint\limits_{[0, 1]^{2n}}   
\widehat{h}(z_1, \dots, z_n) \Phi(z_1, \dots, z_n) \overline{\Psi(z_1, \dots, z_n)} dx_1 \wedge dy_1 \wedge \dots \wedge dx_n \wedge dy_n.
\end{equation}

For any $1 \le j \le m$, Haldane and Rezayi  \cite{HaldaneRezayi} introduced a function $\CC^n \stackrel{\Phi_j}\lar \CC$ given by the following expression 
\begin{equation}\label{E:HRWave}
\Phi_j(z_1, \dots, z_n) := \vartheta\left[\dfrac{j-1}{m}, 0\right](mw + \xi\,|\,m\tau) \left(\prod\limits_{p < q} \vartheta(z_p - z_q)\right)^m,
\end{equation}
where  $w = z_1 + \dots + z_n$.  Let $\mathsf{V}_{m, n, \xi} = \langle \Phi_1, \dots, \Phi_m\rangle_{\CC}$ be their linear span.

\begin{theorem}\label{T:ManybodyTorus} For any $m, n \in \NN$ and $\xi \in \CC$, the following results are true.
\begin{itemize}
\item We have: 
$
\mathsf{V}_{m,n,  \xi} \subset \mathsf{W}_{m,n,  \xi}.
$
In particular, $\mathsf{V}_{m, n,  \xi}$ is a subspace of the ground space of the  self-adjoint operator
\begin{equation}\label{E:nbodytorus}
H_{n}^{tor} = -\dfrac{1}{4} \sum\limits_{k = 1}^n \bigl(\partial_{u_k}^2 + \partial_{v_k}^2 \bigr) - 2\pi i \dfrac{mn}{t} \sum\limits_{k = 1}^n 
v_k \bar\partial_k
\end{equation} 
acting on  the Hilbert space  of global  square integrable sections of the hermitian line bundle  $\kW_{m, n, {\xi}}$ (here, $z_ k = u_k + i v_k$ with $u_k, v_k \in \RR$ for all $1 \le k \le n$). Note that the operator $H_n^{tor}$ is a generalization of the operator $H_{n, 0}$ given by (\ref{E:ClassicMagneticSchroedinger}). 
\item 
Moreover,  $\mathsf{V}_{m,n,  \xi}$ is precisely the space of those holomorphic sections of the line bundle $\kW_{m, n, \xi}$, which have at least order $m$ vanishing along of each partial diagonal $\bigl\{(z_1, \dots, z_n) \in X\, \big| \, z_k = z_l\bigr\}$, $1 \le k \ne l \le n$. 
\item We have: $\langle \Phi_p, \Phi_q\rangle = 0$  and $\langle \Phi_p, \Phi_p\rangle = \langle \Phi_q, \Phi_q\rangle$ for all for $1 \le p \ne q \le m$.
\end{itemize}
\end{theorem}

\noindent
Proofs of the above results can be found in \cite[Section 3]{BurbanK}.

\begin{lemma}\label{L:NormManyBody} For any $1 \le k, l \le m$ we have:
\begin{equation}\label{E:Norms}
\langle \Phi_k, \Phi_l \rangle = \delta_{kl} e^{\frac{2\pi}{m} ta^2}\cdot \gamma,
\end{equation}
where $\gamma \in \RR$ is a constant, which depends on $m, n$ and $\tau$ and is independent of $\xi$ and $l$. 
\end{lemma}

\begin{proof}
Due to Theorem \ref{T:ManybodyTorus}, it suffices to compute $\langle \Phi_1, \Phi_1\rangle$. By definition, we have: 
\begin{equation}
\langle \Phi_1, \Phi_1\rangle=\int_X\big|\vartheta\left[0, 0\right](mw + \xi\,|\,m\tau)\big|^2 \prod\limits_{p < q} \left|\vartheta(z_p - z_q)\right|^{2m}
\widehat{h}(z_1, \dots, z_n) \, \omega,
\end{equation}
where $\omega = dx_1 \wedge dy_1 \wedge \dots \wedge dx_n \wedge dy_n$. 
Consider the shift 
$$
X \stackrel{\iota}\lar X, \bigl(z_1, \dots, z_n\bigr) \mapsto \left(z_1 - \frac{1}{mn} \xi, \dots, 
z_n - \frac{1}{mn} \xi\right).
$$
It is clear that $\iota^\ast(\omega) = \omega$. Moreover, 
$$
\int_X f \kA = \int_X \iota^\ast(f) \omega
$$
for any smooth function $X \stackrel{f}\lar \CC$. It follows that

\begin{align*}
\langle \Phi_1, \Phi_1\rangle&=e^{\frac{2\pi}{m}  ta^2} \iint\limits_{[0, 1]^{2n}}  
\big|\vartheta\left[0, 0\right](mw\,|\,m\tau)\big|^2 \prod\limits_{p < q} \left|\vartheta(z_p - z_q)\right|^{2m}
\prod_{p=1}^nh_{mn,0}(x_p, y_p)\prod_{p=1}^ndx_pdy_p\\&= e^{\frac{2\pi}{m}  ta^2} \cdot \mathcal \gamma,
\end{align*}
where $\gamma \in \RR$ depends only on $m, n \in \NN$ and $\tau \in \CC$.
\end{proof}
Unfortunately, for  $m \ge 2$ we do not have an explicit formula for the constant $\gamma$ appearing in (\ref{E:Norms}). However, for $m = 1$ there is the following result. 

\begin{proposition}\label{P:Slater} For any $n \in \NN$ and $\xi = a \tau + b$, consider the following function:
\begin{equation}
\Phi(z_1, \dots, z_n) = \frac{1}{\sqrt{n!}}
\left|
\begin{array}{ccc}
s_1(z_1) & \dots & s_1(z_n) \\
\vdots & \ddots & \vdots \\
s_n(z_1) & \dots & s_n(z_n) \\
\end{array}
\right|
\end{equation}
where $s_j(z) = \vartheta\bigl[\frac{j-1}{n}, 0\bigr](nz + \xi | n \tau)$ for  $1 \le j \le n$. Then the following statements are true.
\begin{itemize}
\item $\Phi$ is a global holomorphic section of the line bundle $\kW := \kL_{n, \xi} \boxtimes \dots \boxtimes \kL_{n, \xi}$ on the abelian variety $X$. 
\item There exists a constant $\mu \in \CC^\ast$ such that $\Phi(z_1, \dots, z_n) = \mu \cdot \widetilde{\Phi}(z_1, \dots, z_n)$, where
\begin{equation}\label{E:Fay}
\widetilde{\Phi}(z_1, \dots, z_n) =  \vartheta\left[\frac{n-1}{2}, \frac{n-1}{2}\right]\left((z_1 + \dots + z_n) + \xi\,|\,\tau\right) \cdot \prod\limits_{p < q} \vartheta(z_p - z_q).
\end{equation}
\item Finally, we have:
\begin{equation}\label{E:NormSlater}
\lVert \Phi \rVert = \left(\frac{1}{2nt}\right)^{\frac{n}{4}} \cdot e^{\pi t a^2},
\end{equation}
where the scalar product is given by   the formula (\ref{E:ScalarP}) with respect to the metric $\widehat{h}$ given by  the expression  (\ref{E:MetricManyBody}) for $m = 1$. 
\end{itemize}
\end{proposition}

\begin{proof} For any $1 \le j \le n$,  $s_j$ is a global holomorphic section of the line bundle $\kL_{n, \xi}$. Hence, $\Phi$ is indeed a global holomorphic section of the line bundle $\kW$ on  $X$. Computing the automorphy factors of the function 
$\widetilde{\Phi}$  with respect to the shifts $z_k \mapsto z_k +1$ and $z_k \mapsto z_k + \tau$ for all $1 \le k \le n$, we conclude that $\widetilde{\Phi}$ is also a global holomorphic section of  $\kW$. As a consequence, the quotient $\mu:= \dfrac{\Phi}{\widetilde{\Phi}}$ is a meromorpic function on $X$ having at most simple poles. The residue theorem implies that $\mu$ is constant (alternatively, one can use the fact that the vector space $\mathsf{V}_{1, n, \xi}$ is one-dimensional; see Theorem \ref{T:ManybodyTorus}). 

\medskip
\noindent
Finally, the Leibniz formula for the determinant implies that
$$
\langle \Phi, \Phi\rangle = \frac{1}{n!}\sum_{\substack{i_1, \dots, i_n \\ j_1, \dots, j_n}} \varepsilon_{i_1, \dots, i_n} \varepsilon_{j_1, \dots, j_n} \int_X s_{i_1}(z_1) \dots 
s_{i_n}(z_n) \overline{s_{j_1}(z_1)}  \dots 
\overline{s_{j_n}(z_n)} h(z_1) \dots h(z_n) \omega,
$$
where the sum is taken over all permutations 
$\left(\begin{array}{ccc}
1 & \dots & n \\
i_1 & \dots & i_n
\end{array} \right), 
\left(\begin{array}{ccc}
1 & \dots & n \\
j_1 & \dots & j_n
\end{array} \right)
$
and $\varepsilon_{i_1, \dots, i_n}, \varepsilon_{j_1, \dots, j_n}$ are the corresponding  signatures. It follows that
\begin{equation}\label{E:SumIntermediate}
\langle \Phi, \Phi\rangle = \frac{1}{n!} \sum_{\substack{i_1, \dots, i_n \\ j_1, \dots, j_n}} \varepsilon_{i_1, \dots, i_n} \varepsilon_{j_1, \dots, j_n} \langle s_{i_1}, s_{j_1}\rangle \dots \langle s_{i_n}, s_{j_n}\rangle. 
\end{equation}
By  Proposition \ref{P:ScalarProductTheta}, $(s_1, \dots, s_n)$ is  an orthonormal basis of the space $\mathsf{W}_{n, \xi}$ with respect to the hermitian product (\ref{E:HermitianProduct}). Hence, the only non-zero summands in (\ref{E:SumIntermediate}) are those for which $(i_1, \dots, i_n) = (j_1, \dots, j_n)$. By Proposition \ref{P:ScalarProductTheta},  they are all  equal to 
$$\left(\sqrt{\frac{1}{2nt}} 
e^{2\pi t \frac{a^2}{n}}\right)^n = \left(\frac{1}{2nt}\right)^{\frac{n}{2}} \cdot e^{2 \pi t a^2}$$
and there are $n!$ such summands. This implies the formula (\ref{E:NormSlater}). 
\end{proof}

\begin{remark}
The formula \eqref{E:Fay} can also be found in \cite{Fay92}; see (5.33) loc. cit. 
\end{remark}

\section{Multi-layer model of FQHE on a torus}\label{S:WaveFunctions}

\noindent
Following \cite{KeskiVakkuriWen, Wen}, 
the the multi-layer torus model of FQHE is defined as follows \cite[Definition 4.1]{BurbanK}. 

\begin{definition}\label{D:Kmatrix}
Let $K \in \Mat_{g \times g}(\NN_0)$ be a matrix satisfying the following conditions
\begin{itemize}
\item $K$ is symmetric and positive definite.
\item All diagonal entries of $K$ are either even ($\epsilon(K) := 1$,  bosonic case) or odd ($\epsilon(K) := -1$, fermionic case). 
\item All entries of the vector $\vec{u} := K^{-1} \vec{e} \in \mathbb{Q}^g$
are positive, where $\vec{e} = \left(\begin{array}{c} 1 \\ \vdots \\ 1\end{array}\right)$. 
\end{itemize}
We put $\Pi := \Pi_K = K^{-1} \ZZ^g/\ZZ^g$. Note that $\Pi$ is a finite abelian group of order $\delta = \det(K)$ isomorphic to  $\ZZ^g/ K \ZZ^g$.
\end{definition}

For any symmetric matrix $\Omega \in \Mat_{g \times g}(\CC)$, whose  imaginary part $\mathfrak{Im}(\Omega)$ is positive definite we define for any 
$\vec{z} \in \CC^g$ and $\vec{a}, \vec{b} \in \RR^g$ (analogously to  (\ref{E:ThetaWithCharacteristics})   the following series:
\begin{equation}\label{E:MultiThetaWithCharacteristics}
\Theta[\vec{a}, \vec{b}](\vec{z}\,|\, \Omega):= \sum\limits_{\vec{k} \in \ZZ} 
\exp\bigl(\pi i (\vec{k} + \vec{a})^t \Omega (\vec{k} + \vec{a})
+  2\pi i (\vec{k} + \vec{a})^t (\vec{z}+\vec{b})\bigr).
\end{equation}

\begin{definition}\label{D:WenData} Let $K\in \Mat_{g \times g}(\NN_0)$ be a matrix satisfying the constraints of  Definition \ref{D:Kmatrix}. Next, let  
$\vec{\mathsf{n}} = \left(\begin{array}{c} 
n_1 \\ \vdots \\ n_g\end{array}\right) \in \NN^g$ and 
$d \in \NN$ be such that $K \vec{\mathsf{n}} = d \vec{e}$. 
Then we put: 
\begin{itemize}
\item $X = \underbrace{E \times \dots \times E}_{n_1  \, \mbox{\scriptsize{\sl times}}} \times \dots \times \underbrace{E \times \dots \times E}_{n_g  \, \mbox{\scriptsize{\sl times}}}$ and $n = n_1 + \dots + n_g = \mathsf{dim}(X)$. 
\item $D_K = 
\prod\limits_{k = 1}^g \left(\prod\limits_{1 \le p < q \le n_k} 
\vartheta\left(z_p^{(k)} - z_q^{(k)} \right)\right)^{K_{kk}} \cdot \prod\limits_{ 1 \le k < l \le g} \left(\prod\limits_{p = 1}^{n_k} \prod\limits_{q = 1}^{n_l}
\vartheta\left(z_p^{(k)} - z_q^{(l)} \right)\right)^{K_{kl}},
$
where $\left(z_1^{(1)}, \dots, z_{n_1}^{(1)}; \dots; z_1^{(g)}, \dots, z_{n_g}^{(g)}\right)$ are the standard local coordinates on  $X$. 
\item For any $1 \le k \le g$ we put: $w_k = z_1^{(k)} + \dots + z_{n_k}^{(k)}$ and  $\vec{w} = (w_1, \dots, w_g)$. 
\item For any $\vec{c} \in \Pi$ and $\vec\zeta \in \CC^g$, the  corresponding wave function of Keski-Vakkuri and Wen \cite{KeskiVakkuriWen} is given by the expression 
\begin{equation}\label{E:KV-Wen-wavefunct}
\Phi_{\vec{c}}\Bigl(z_1^{(1)}, \dots, z_{n_1}^{(1)}; \dots;  z_1^{(g)}, \dots, z_{n_g}^{(g)}\Bigr) =  \Theta\bigl[\vec{c}, \vec{0}\bigr](K \vec{w} + \vec{\zeta} | \Omega) \cdot D_K,
\end{equation}
where $\Omega = \tau K$. 
\end{itemize}
Following \cite{Wen}, we call the pair $(K, \vec\nn)$ a \emph{Wen datum}, which specifies the multi-layer torus model of FQHE 
and $\mathsf{V}_{K, \vec{\nn}, \vec{\zeta}} = \bigl\langle 
\Phi_{\vec{c}} \, | \;  \vec{c} \in \Pi \bigr\rangle_{\CC}$ is the space  of multi-layer wave functions of Keski-Vekkuri and Wen.
\end{definition}

\noindent
Next, we put 
\begin{equation}\label{E:linebundleManybody}
\kW_{K, \vec{\nn}, \vec{\zeta}} := \left\{
\begin{array}{cl}
\underbrace{\kL_{d, \zeta_1} \boxtimes \dots \kL_{d, \zeta_1}}_{n_1  \; \mbox{\scriptsize{\sl times}}} \boxtimes \dots \boxtimes \underbrace{\kL_{d, \zeta_g} \boxtimes \dots \boxtimes \kL_{d, \zeta_g}}_{n_g  \; \mbox{\scriptsize{\sl times}}} 
& \mbox{\rm if} \, \epsilon(K) + d \in 2 \NN \\
\underbrace{\kL^\sharp_{d, \zeta_1} \boxtimes \dots \kL^\sharp_{d, \zeta_1}}_{n_1  \; \mbox{\scriptsize{\sl times}}} \boxtimes \dots \boxtimes \underbrace{\kL^\sharp_{d, \zeta_g} \boxtimes \dots \boxtimes \kL^\sharp_{d, \zeta_g}}_{n_g  \; \mbox{\scriptsize{\sl times}}} & \mbox{\rm otherwise}
\end{array}
\right. 
\end{equation}
and set $\mathsf{W}_{K, \vec{\nn}, \vec{\zeta}} = \Gamma\bigl(X,  \kW_{K, \vec{\nn}, \vec{\zeta}}\bigr)$. 
As before for the one-layer Laughlin states, the hermitian metric on $\kW_{K, \vec{\nn}, \vec{\zeta}}$ is the point-wise product metric 
$\widehat{h}\bigl(\{z^{k}_p\}\bigr) := \prod_{k=1}^g\prod_{p=1}^{n_k}\otimes h_{d, \zeta_k}(z_p^{(k)})$ induced by the metric \eqref{E:MetricLB1}.
Again, on each copy of $E$ in $X$ we consider the volume form $dx_p^{(k)}\wedge dy_p^{(k)}$, which gives a volume form of $X$.  Then the natural inner product for the sections $\Phi,\Psi\in \mathsf{W}_{K, \vec{\nn}, \vec{\zeta}}$ reads
\begin{equation}\label{E:ScalarProductMultilayer}
\langle \Phi, \Psi\rangle=\int_X\Phi\big(\{z_p^{(k)}\}\big)\overline{\Psi\big(\{z_p^{(k)}\}\big)}\prod_{k=1}^g\prod_{p=1}^{n_k}h(z_p^{(k)})\prod_{k=1}^g\prod_{p=1}^{n_k}dx_p^{(k)}dy_p^{(k)}.
\end{equation}

\begin{lemma}\label{L:WenDatumCyclic} For any Wen datum $(K, \vec\nn)$ we have: $\frac{n\delta}{d} \in \ZZ$ and $\vec{u} \in K^{-1} \ZZ^g$. Moreover, if $\mathsf{gcd}\bigl(\delta,  \frac{n\delta}{d}\bigr) = 1$ then the class of $\vec{u}$ generates the group $\Pi$.
\end{lemma}

\begin{proof}  It follows from the definition that   $$\vec{u} = \frac{1}{d}\left(\begin{array}{c} n_1 \\ \vdots \\ n_g\end{array}\right) = 
K^{-1} \vec{e} = \frac{1}{\delta} K^\sharp \vec{e},
$$
where $K^\sharp$ is the adjunct matrix of $K$. Since $K^\sharp \in \Mat_{g \times g}(\ZZ)$, we see that $\vec{u} \in K^{-1} \ZZ^g$ and $\delta \vec{u} \in \ZZ^g$. Hence, 
$\frac{n\delta}{d} \in \ZZ$, as asserted.

Assume now that $\mathsf{gcd}\bigl(\delta,  \frac{n\delta}{d}\bigr) = 1$. Since $|\Pi| = \delta$, the order of $[\vec{u}] \in \Pi$ is a divisor of $\delta$. If $\delta = s \bar{\delta}$ with $s > 1$ and $\bar\delta \vec{u} \in \ZZ^g$ then $\frac{n_i \bar\delta}{d} \in \ZZ$ for all $1\le i \le g$. It follows that 
$\frac{n\bar\delta}{d} \in \ZZ$ and $s \big| \mathsf{gcd}\bigl(\delta,  \frac{n\delta}{d}\bigr)$, yielding a contradiction. 
\end{proof}

\begin{example} For any $p, g \in \NN$ consider the following matrix
\begin{equation}
K = 
\left(
\begin{array}{cccc}
p+1 & p & \dots & p \\
p & p+1 & \dots & p \\
\vdots & \vdots & \ddots & \vdots \\
p & p & \ddots & p+1
\end{array}
\right)
\in \mathsf{Mat}_{g \times g}(\NN).
\end{equation}
Then $K$ has precisely two eigenvalues: $1$ (with multiplicity $g-1$) and $\delta = pg +1$ (the corresponding eigenvector is $\vec{e}$). Hence, $K$ is positive definite and its determinant is $\delta$. We can take 
$\vec{n}^t=(m,...,m)$ (hence, $n = mg$). It follows that  $d= m \delta$ and 
$\frac{n \delta}{d} = g$ is coprime to $pg +1$. Hence, conditions of Lemma \ref{L:WenDatumCyclic} are satisfied. 
\end{example}

\begin{definition} Let $K$ be a matrix as in Definition \ref{D:Kmatrix}.
\begin{itemize}
\item For any $1 \le k \le g$ and $1 \le p < q \le n_k$ we put:
$$
\Xi^{(k)}_{p, q} := \left\{\bigl(z_1^{(1)}, \dots, z_{n_1}^{(1)};  \dots; z_1^{(g)}, \dots, z_{n_g}^{(g)}\bigr) \in X \, \big| \, z_p^{(k)} = z_q^{(k)} \right\}.
$$
\item Similarly, for any $1 \le k < l \le g$, $1 \le p \le n_k$  and $1 \le q \le n_l$ we put:
$$
\Xi_{p, q}^{(k, l)} := \left\{\bigl(z_1^{(1)}, \dots, z_{n_1}^{(1)};  \dots; z_1^{(g)}, \dots, z_{n_g}^{(g)}\bigr) \in X \, \big| \, z_p^{(k)} = z_q^{(l)} \right\}.
$$
\item Then we have  the  following divisor 
\begin{equation*}
\Xi_K:= \sum\limits_{k = 1}^g K_{kk}\sum\limits_{1 \le p < q \le n_k} \left[\Xi_{p, q}^{(k)}\right] + 
\sum\limits_{1 \le k < l \le g} K_{kl} \sum\limits_{p = 1}^{n_k} \sum\limits_{q = 1}^{n_l} \left[\Xi_{p, q}^{(k, l)}\right]
\end{equation*}
on the abelian variety $X$. 
\end{itemize}
\end{definition}

\noindent
Theorem \ref{T:ManybodyTorus} has the following generalization.

\begin{theorem}\label{T:MultiLayerTorus} For Wen datum $(K, \vec\nn)$ and $\vec\zeta \in \CC^g$, the following results are true.
\begin{itemize}
\item We have: 
$
\mathsf{V}_{K, \vec{\mathsf{n}},  \zeta} \subset \mathsf{W}_{K, \vec{\mathsf{n}},  \zeta}.
$
In particular,  $\mathsf{V}_{K, \vec{\mathsf{n}}, \vec{\zeta}}$ is a subspace of the ground space of an approproate many-body magnetic Schr\"odinger  operator acting on the space of global $L^2$-sections of the hermitian holomorphic line bundle $\kW_{d, \vec{\mathsf{n}}, \vec\zeta}$.  
\item 
Moreover,  $\mathsf{V}_{K,n,  \xi}\cong \Gamma\bigl(X, \kW_{K, \vec{\mathsf{n}}, \vec\zeta}(-\Xi_K)\bigr)$ and $\dim_{\CC}(\mathsf{V}_{K,n,  \xi}) = \delta$. 
\item Assume that $(K, \vec\nn)$ satisfies the conditions of Lemma \ref{L:WenDatumCyclic}. For any $1\le p \le \delta$ we put: $\Phi_p = \Phi_{(p-1)\vec{u}}$. Then we have: $\langle \Phi_p, \Phi_q\rangle = 0$  and $\langle \Phi_p, \Phi_p\rangle = \langle \Phi_q, \Phi_q\rangle$ for all for $1 \le p \ne q \le \delta$.
\end{itemize}
\end{theorem}

\noindent
Proofs of the above results can be found in \cite[Section 4]{BurbanK}.

\smallskip
\noindent
We have the following generalization of Lemma \ref{L:NormManyBody}.

\begin{lemma}\label{L:NormMultiLayer} Assume that $(K, \vec\nn)$ satisfies the conditions of Lemma \ref{L:WenDatumCyclic} and $\vec\zeta = \xi \vec{e}$ for $\xi = a\tau + b$. Then for any $1 \le k, l \le \delta$ we have:
\begin{equation}\label{E:NormsMultiL}
\langle \Phi_k, \Phi_l \rangle = \delta_{kl} e^{\frac{2\pi n}{d} ta^2}\cdot \gamma,
\end{equation}
where $\gamma \in \RR$ is a constant, which depends on $(K, \vec\nn)$  and $\tau$ and is independent of $\xi$ and $l$. 
\end{lemma}

\begin{proof}
The argument is the same as in the proof of Lemma \ref{L:NormManyBody}. 
Due to the previous theorem,  it is enough to compute $\langle \Phi_1, \Phi_1\rangle$. We  have: 
\begin{align}\label{E:IntegralMultiL}
\langle \Phi_1, \Phi_1\rangle=\int_X\big|\Theta\bigl[0, \vec{0}\bigr](K\vec{w} + \xi\vec{e}|\Omega)\big|^2\cdot|D_K|^2 \widehat{h}\bigl(z_p^{(k)}\bigr)\prod_{k=1}^g\prod_{p=1}^{n_k}dx_p^{(k)}dy_p^{(k)}.
\end{align}
For any $1 \le k \le g$ and $1 \le p \le n_k$ consider the shift $z^{(k)}_p \mapsto 
z^{(k)}_p - \frac{1}{d} \xi$. Since we integrate over a torus,  the integral (\ref{E:IntegralMultiL}) remains invariant under this shift. On the other hand
\begin{itemize}
\item For any $1 \le k \le g$ we have: $w_k \mapsto w_k - \frac{n_k}{d}\xi$. Moreover, 
$\Theta\bigl[0, \vec{0}\bigr](K\vec{w} + \xi\vec{e}|\Omega) \mapsto \Theta\bigl[0, \vec{0}\bigr](K\vec{w}|\Omega)$; 
\item The discriminant $D_K$ remains invariant under this transformation. 
\item $\widehat{h}\bigl(z_p^{(k)}\bigr) \mapsto e^{\frac{2\pi n}{d} ta^2}\cdot \prod_{k=1}^g\prod_{p=1}^{n_k}h_{d,0}(x_p^{(k)},y_p^{(k)})$. 
\end{itemize}
It follows that
$$
\langle \Phi_1, \Phi_1\rangle= e^{\frac{2\pi n}{d} ta^2}\cdot
\int_X\big|\Theta\bigl[0, \vec{0}\bigr](K\vec{w}|\Omega)\big|^2\cdot|D_K|^2 \prod_{k=1}^g\prod_{p=1}^{n_k}h_{d,0}\bigl(x_p^{(k)},y_p^{(k)}\bigr) \prod_{k=1}^g\prod_{p=1}^{n_k}dx_p^{(k)}dy_p^{(k)}
$$
implying the statement. 
\end{proof}

\section{Center-of-mass Hermitian structure for the multi-layer states}
Let $K \in \Mat_{g \times g}(\ZZ)$ be a symmetric positive definite matrix and $\vec\xi = \tau \vec{a} + \vec{b} \in \CC^g$ with $\vec{a}, \vec{b} \in \RR^g$.
\begin{equation}\label{E:HolomSec}
\mathsf{W}_{K, \vec\xi} :=  \left\{
\CC^g \stackrel{H}\lar \CC \; \mbox{\rm holomorphic}\left| \, 
\begin{array}{l} 
\begin{array}{l}
H(\vec{z}+ \vec{l}) = H(\vec{z}) \\
H(\vec{z}+ \tau \vec{l}) = 
U_{\vec{l}}(\vec{z}, \vec{\xi})  H(\vec{z}) 
\end{array}
 \mbox{\rm for all} \; \vec{l} \in \ZZ^g\\
\end{array}
\right.
\right\},
\end{equation}
where $U_{\vec{l}}(\vec{z}, \vec{\xi}) =  \exp\bigl(-\pi i (\vec{l}, 2 \vec{\xi} + 2K \vec{z} + \Omega \vec{l})\bigr)$ and $\Omega = \tau K$. 
By \cite[Theorem 4.6]{BurbanK}, we know that the dimension of $\mathsf{W}_{K, \vec\xi}$ is $\delta=\det(K)$. Let $\Pi = K^{-1}\ZZ^g/\ZZ^g$.  Then $\big|\Pi\big| = \delta$ and $\mathsf{W}_{K, \vec\xi}$ has a distinguished basis  $\bigl(H_{\vec{c}}\bigr)_{\vec{c} \in \Pi}$, where 
\begin{equation}\label{E:CentralWaveFunctions}
H_{\vec{c}}(\vec{z}) := \Theta\bigl[\vec{c}, \vec{0}\bigr](K \vec{z} + \vec{\xi} | \Omega).
\end{equation}
In what follows, we identify $\vec{c} \in K^{-1}\ZZ^g$ with its class in $\Pi$. 
The vector space $\mathsf{W}_{K, \vec\xi}$ is equipped with natural hermitian inner product,
\begin{equation}\label{E:scalarproductmv}
\langle \Phi_1, \Phi_2 \rangle := \iint\limits_{[0, 1]^{2g}} 
\exp\bigl(-2\pi  t(\vec{y}, K \vec{y} + 2\vec{a})\bigr)
 \Phi_1(\vec{x}, \vec{y})  \overline{\Phi_2(\vec{x}, \vec{y})}  d\vec{x} \wedge d \vec{y}.
\end{equation}
and by \cite[Theorem  4.6]{BurbanK}, for any any $\vec{c}_1 \ne \vec{c}_2 \in \Pi$ we have:
$
\bigl\langle H_{\vec{c}_1}, H_{\vec{c}_2}\bigr\rangle = 0$ and $
\bigl\langle H_{\vec{c}_1}, H_{\vec{c}_1}\bigr\rangle = \bigl\langle H_{\vec{c}_2}, H_{\vec{c}_2}\bigr\rangle.
$
This means that the Gram matrix for the basis \eqref{E:CentralWaveFunctions} is a scalar matrix. 

\begin{proposition}\label{P:NormCenterMass} We have the following formula for the Gram matrix of \eqref{E:CentralWaveFunctions}: 
\begin{equation}\bigl\langle H_{\vec{c}_1}, H_{\vec{c}_2}\bigr\rangle= \kappa(\vec\xi) \cdot \delta_{\vec{c}_1,\vec{c}_2},
\end{equation} 
where $\kappa(\vec\xi)= \left(2t\delta\right)^{-\frac{g}{2}}\cdot 
e^{2\pi t (\vec{a},K^{-1} \vec{a})}
$
\end{proposition}
\begin{proof}
The proof is a straightforward computation  analogous to Proposition  \ref{P:ScalarProductTheta}.

\begin{align}\nonumber
\bigl\langle H_{\vec{c}_1}, H_{\vec{c}_2}\bigr\rangle &=\iint\limits_{[0, 1]^{2g}} e^{-2\pi t(\vec{y},K\vec{y}+2\vec{a})}
\Theta\Bigl[\vec{c}_1, \vec{0}\Bigr](K\vec{z} + \vec{\xi} | K\tau)\overline{\Theta\Bigl[\vec{c}_2, \vec{0}\Bigr](K\vec{z} + \vec{\xi} | K \tau)}
 d\vec{x} \wedge d\vec{y}\\\nonumber
&=
\sum_{\vec{m},\vec{n}\in\mathbb Z^g}e^{\pi i\big(\vec{m}+\vec{c}_1,K\tau(\vec{m}+\vec{c}_1)\big)-\pi i\big(\vec{n}+\vec{c}_2,K\bar\tau(\vec{n}+\vec{c}_2)\big)}
 S_{\vec{m}, \vec{n}},\\ \nonumber
\end{align}
where $$
S_{\vec{m}, \vec{n}} = \iint\limits_{[0, 1]^{2g}}  G_{\vec{m}, \vec{n}}(\vec{x}) H_{\vec{m}, \vec{n}}(\vec{y}) d\vec{x} \wedge d\vec{y}
$$
with 
$$
\begin{array}{l}
G_{\vec{m}, \vec{n}}(\vec{x}) = e^{2\pi i \bigl(K\vec{x}, (\vec{m} - \vec{n}) + (\vec{c_1} - \vec{c_2})\bigr)} \\
H_{\vec{m}, \vec{n}}(\vec{y}) 
= e^{-2\pi t \bigl(\vec{y}, K \vec{y} + 2\vec{a}\bigr)} e^{2\pi i \bigl(\vec{m} + \vec{c}_1, K\tau \vec{y} + \vec{\xi}\bigr)  - 2\pi i \bigl(\vec{n} + \vec{c}_2, K \bar\tau \vec{y} + \vec{\xi}\bigr)}.
\end{array}
$$
Since $K = K^t$, we  get: $\bigl(K\vec{x}, (\vec{m} - \vec{n}) + (\vec{c_1} - \vec{c_2})\bigr) = (\vec{x}, \vec{l})$, where $\vec{l} = K\left((\vec{m} - \vec{n}) + (\vec{c_1} - \vec{c_2})\right)$. As $\vec{c}_1, \vec{c}_2 \in K^{-1} \ZZ^g$, we have: $\vec{l} \in \ZZ^g$. Moreover, $\vec{l} = \vec{0}$ if and only if 
$\vec{m} = \vec{n}$ and $\vec{c_1} =   \vec{c_2} \; \mbox{\rm in} \; \Pi$. 
It follows that
$$
\iint\limits_{[0,1]^g} G_{\vec{m}, \vec{n}} d \vec{x} = 
\prod\limits_{k = 1}^g 
\int\limits_{0}^1 e^{2\pi i l_k x_k} dx_k = 
\left\{
\begin{array}{cl}
1
& \mbox{\rm if} \;  \vec{m} = \vec{n} \;  \mbox{and} \;   \vec{c_1} =   \vec{c_2} \; \mbox{\rm in} \; \Pi \\
0
 & \mbox{\rm otherwise}.
\end{array}
\right. 
$$
As a consequence, we have:  $\bigl\langle H_{\vec{c}_1}, H_{\vec{c}_2}\bigr\rangle = 0$ if $\vec{c_1} \ne \vec{c_2}$ in $\Pi$. 

\smallskip
\noindent
Now, assume that $\vec{c_1}  = \vec{c}_2 = \vec{c}$. Since $S_{\vec{m}, \vec{n}} = 0$ for all $\vec{m} \ne  \vec{n}$,  we have: 
\begin{align}\nonumber
\bigl\langle H_{\vec{c}}, H_{\vec{c}}\bigr\rangle &=\sum_{\vec{n}\in\mathbb Z^g}e^{-2\pi t\big(\vec{n}+\vec{c},K(\vec{n}+\vec{c})\big)}\iint\limits_{[0, 1]^{g}} e^{-2\pi t(\vec{y},K\vec{y}+2\vec{a})-4\pi t(\vec{n}+\vec{c},K\vec{y}+\vec{a})}d\vec{y}\\\nonumber
&= e^{2\pi t (\vec{a}, K^{-1}\vec{a})} \sum_{\vec{n}\in\mathbb Z^g}\iint\limits_{[0, 1]^{g}} e^{-2\pi t\big(\vec{y}+\vec{n}+\vec{c}+K^{-1}\vec{a},K(\vec{y}+\vec{n}+\vec{c})+\vec{a}\big)}d\vec{y}\\\nonumber
&=e^{2\pi t (\vec{a}, K^{-1}\vec{a})}\iint\limits_{\mathbb R^{g}} e^{-2\pi t\bigl(\vec{y} + \vec{c} + K^{-1} \vec{a}, K (\vec{y} + \vec{c} + K^{-1} \vec{a})\bigr)}d\vec{y}
\\\nonumber
&=e^{2\pi t (\vec{a}, K^{-1}\vec{a})}\iint\limits_{\mathbb R^{g}} e^{-2\pi t(\vec{y}, K\vec{y})}d\vec{y} =(2t\delta)^{\frac g2}\cdot e^{2\pi t(\vec{a},K^{-1}\vec{a})}. \\\nonumber
\end{align}
Proposition is proven.
\end{proof}

\section{Hermitian connections and adiabatic curvature in the fractional quantum Hall effect}

We  first recall the definition of the canonical Bott--Chern connection of a holomorphic hermitian vector bundle on a complex manifold; see \cite{BottChern, Kobayashi}.

Let $M$ be a smooth real manifold, $\kO$ the corresponding structure sheaf, $\kA^p$ the sheaf of differential $p$-forms on $M$ and $\kE$ a smooth complex vector bundle of rank $n$ on $M$. Recall that a connection on $\kE$ is a morphism of sheaves of complex vector spaces $\kE \stackrel{\nabla}\lar \kA^1 \otimes \kE$ such that for any 
open subset $U \subset M$, $f \in \Gamma(U, \kO)$ and $s \in \Gamma(U, \kE)$ we have: 
$$
\nabla(f \cdot s) = f \cdot \nabla(s) + df \otimes s. 
$$
We get  an induced morphism of sheaves $\kA^1 \otimes \kE \stackrel{\widetilde\nabla}\lar \kA^2 \otimes \kE$ given locally by the formula
$$
\widetilde{\nabla}(\theta \otimes s) = d\theta \otimes s - \theta \wedge \nabla(s),
$$
where $\theta \in \Gamma(U, \kA^1)$, $s \in \Gamma(U, \kE)$ and  $\theta \wedge \nabla(s) = \sum\limits_{k = 1}^m (\theta \wedge \theta_k) \otimes s_k$ if $\nabla(s) = \sum\limits_{k = 1}^m \theta_k \otimes s_k$. It can be checked that the  composition $\widetilde{\nabla}\nabla: \kE \lar \kA^2 \otimes \kE$ is a morphism of complex vector bundles. Using the canonical isomorphism
$\Hom_M(\kE, \kA^2 \otimes \kE) \lar \Gamma\bigl(M, \kA^2 \otimes \mathit{End}_\kO(\kE)\bigr)$, we get a distinguished section $K_\nabla \in \Gamma\bigl(M, \kA^2 \otimes \mathit{End}_\kO(\kE)\bigr)$ called curvature of $\kE$. A connection $\nabla$ is called flat if $K_\nabla = 0$ and \emph{projectively flat} if there exists
$\kA \in \Gamma(M, \kA^2)$ such that $K_\nabla = \omega \cdot  \mathsf{Id}_\kE$. 

 Let $U \subset M$ be an open subset and $\psi_1, \dots, \psi_n \in \Gamma(U, \kE)$. Then 
$F = (\psi_1, \dots, \psi_n)$ is a local frame of $\kE$ if $\psi_1(p), \dots, \psi_n(p) \in \kE\big|_p$ form a basis of $\kE\big|_p$ for all $p \in U$. A choice of a frame $F$ is equivalent to a choice of a local trivialization $\kO_U^n \lar \kE\big|_U$.

Given such a frame $F$, we can find for any $1 \le k, l \le n$ a unique $\theta_{kl} \in \Gamma(U, \kA^1)$ such that
$
\nabla(\psi_k) = \sum\limits_{l = 1}^n \theta_{kl} \psi_l. 
$
Next, we put:
$
\kA_{kl} := d\theta_{kl} - \sum\limits_{p = 1}^n \theta_{kp} \wedge \theta_{pl}.
$
The $(n\times n)$-matrix $K^F_\nabla$ gives a local description of the curvature 
$K_\nabla$ with respect to the frame $F$. Its trace
\begin{equation}\label{E:ConnectionTrace}
\omega_{\nabla} = \mathsf{tr}(K^F_\nabla) = d\Bigl( \sum\limits_{p = 1}^n \theta_{pp}\Bigr) \in \Gamma(U, \kA^2)
\end{equation}
does not depend on the choice of a frame $F$. Because of this invariance, we actually get a global closed $2$-form $\omega_\nabla \in \Gamma(M, \kA^2)$, whose local description is given by  (\ref{E:ConnectionTrace}). 

Next, assume that $h = \langle -\,,\, -\rangle$ is a hermitian metric on $\kE$. Then a connection $\kE \stackrel{\nabla}\lar \kA^1 \otimes \kE$ is hermitian (or metric)  with respect to $h$ if
$$
d\langle s_1, s_2\rangle = \bigl\langle \nabla(s_1), s_2\bigr\rangle + 
\bigl\langle s_1, \nabla(s_2)\bigr\rangle
$$
for any open subset $U \subseteq M$ and any local sections $s_1, s_2 \in \Gamma(U, \kE)$. 

Now assume that $M$ is a complex manifold of complex dimension $d$ and $\kE$ is a holomorphic vector bundle on $M$. For any $p, q \in \NN_0$, we have the corresponding sheaf of differential $(p, q)$-forms $\kA^{p, q}$ as well as sheaf morphisms  $\kA^{p, q} \stackrel{\partial} \lar \kA^{p+1, q}$ and $\kA^{p, q} \stackrel{\bar\partial} \lar \kA^{p, q+1}$ given by the formulae
$$
\partial(w) = \sum\limits_{l = 1}^d 
\frac{\partial f_{IJ}}{\partial z_l} dz_l \wedge dz_I \wedge d\bar{z}_J
\quad \mbox{\rm and} \quad 
\bar{\partial}(w) = \sum\limits_{l = 1}^d
\frac{\partial f_{IJ}}{\partial \bar{z}_l} d\bar{z}_l \wedge dz_I \wedge d\bar{z}_J,
$$
where 
$$
w = \sum\limits_{I, J} f_{IJ} dz_I \wedge d\bar{z}_J = 
\sum_{\substack{i_1 < \dots < i_p \\ j_1 < \dots < j_q}}  f_{i_1, \dots, i_p; j_1, \dots j_q} dz_{i_1} \wedge \dots \wedge dz_{k_p} \wedge d \bar{z}_{j_1} \wedge \dots \wedge  \bar{z}_{j_q} \in \Gamma\bigl(U, \kA^{p, q}\bigr).
$$
In particular, there is  a direct sum decomposition $\kA^1 = \kA^{0, 1} \oplus \kA^{1, 0}$. A holomorphic structure on $\kE$ defines a morphism of sheaves $ \kE \stackrel{\bar\partial_{\kE}}\lar \kA^{0, 1} \otimes \kE$. 
By a result of Bott and Chern \cite[Proposition 3.2]{BottChern} for any complex manifold $M$ and a holomorphic hermitian vector bundle $(\kE, h)$ there exists a  \emph{unique} hermitian connection 
$\kE \stackrel{\nabla}\lar  \kA^1 \otimes \kE$ such that
$\nabla^{(0, 1)} = \bar\partial_{\kE}$. 

The curvature $K_\nabla$ of this canonical connection $\nabla$ admits the following local description. For an open subset $U \subseteq M$, let $F = (\psi_1, \dots, \psi_n)$ be a \emph{holomorphic} frame of $\kE$. Consider the corresponding  Gram matrix
$C = (\langle \psi_k, \psi_l\rangle) \in \mathsf{Mat}_{n \times n}\bigl(\Gamma(U, \kO)\bigr)$. 
Then the curvature matrix of $\nabla$ with respect to the frame $F$ is given by the following expression: 
\begin{equation}\label{E:CurvatureBottChern}
K^F_\nabla = \bar\partial \bigl(\partial(C) \cdot C^{-1}\bigr).
\end{equation}

\smallskip
\noindent
Now, let $(K, \vec{\mathsf{n}})$ be a Wen datum.  In \cite{BurbanK}, we constructed a holomorphic hermitian line bundle $\kB$ on the torus $E = \CC/\langle 1, \tau\rangle$
having the following property: for any $\xi \in \CC$  there exists an  isomorphisms  of \emph{hermitian} complex  vector spaces $\kB\Big|_{[\xi]} \cong \mathsf{V}_{K, \vec{\nn}, \xi}$, where the hermitian product on  $\mathsf{V}_{K, \vec{\nn}, \xi}$ is the restriction of the  inner product (\ref{E:ScalarProductMultilayer}) on the Hilbert space of the corresponding many-body system. In other words, $\kB$ has rank $\delta$ and $F:= \bigl(\Phi_1, \dots, \Phi_\delta)$ is a \emph{holomorphic} frame of  $\kB$ in a neighbourhood of the point $[\xi] \in \kE$. 

\begin{theorem}\label{T:CurvatureLaughlinBundle} 
Let $(K, \vec{\mathsf{n}})$ be a Wen datum satisfying the conditions of Lemma \ref{L:WenDatumCyclic}. 
Then the  Bott--Chern connection $\nabla$ of the hermitian vector bundle $\kB$ is projectively flat. Moreover, the trace of the curvature of $\nabla$ is $\omega_\nabla   = - \dfrac{\pi}{t} \cdot  \dfrac{n \delta}{d} d\xi \wedge d\bar\xi$ (recall that $\xi = a\tau + b$ with $a, b \in \RR$) and  $\mathsf{deg}(\kB) = - \dfrac{n \delta}{d}$. 
\end{theorem}
\begin{proof} According to Lemma \ref{L:NormMultiLayer}, the Gram matrix $C$ of $\nabla$ with respect to the frame $F$ is a scalar matrix. The formula (\ref{E:CurvatureBottChern}) implies that 
$
K^F_\nabla = \omega I_{g \times g},
$
where $$
\omega = \bar\partial \partial \log\left(e^{\frac{2\pi n}{d} ta^2}\cdot \gamma\right) = 
\bar\partial \partial\left(\frac{2\pi n}{d} ta^2 \right).
$$
In particular, 
$\nabla$ is projectively flat, as asserted. Since 
$a = \dfrac{\xi -\bar\xi}{2it}$, we have: 
$\dfrac{2\pi n}{d} ta^2 = - \dfrac{\pi n}{2t d}(\xi^2 - 2\xi \bar\xi + \bar\xi^2)$. It follows that
\begin{equation}\label{E:HallConductance}
\omega_\nabla  = \delta \cdot \omega = - \frac{\pi}{t} \cdot  \frac{n \delta}{d} d\xi \wedge d\bar\xi. 
\end{equation}
In particular, $c_1(\kB) = \dfrac{i}{2\pi} \omega_\nabla = - \dfrac{i}{2t} \cdot \dfrac{n \delta}{d}d\xi \wedge d\bar\xi$ and $\mathsf{deg}(\kB) = - \dfrac{n \delta}{d}$. 
\end{proof}

\begin{remark} It follows that the absolute value of the slope of $\kB$ is 
$$
\left| 
\dfrac{\mathsf{deg}(\kB)}{\mathsf{rk(\kB)}}
\right| = \frac{n}{d}.
$$
From the physical point of view, the $2$-form (\ref{E:HallConductance}) describes the Hall conductance of the electron gas moving a torus in a presence of a uniform magnetic field. Using the technique of Fourier-Mukai transforms, we proved in our previous work  \cite{BurbanK}  that the holomorphic vector bundle $\kB$ is simple. Moreover, a more general magnetic vector bundle of the abelian variety $\CC^g/(\ZZ^g + \tau \ZZ^g)$ was constructed and investigated. 
\end{remark}

\begin{remark} One can formulate algebro-geometric models of FQHE on Riemann surfaces of genus $\ge 2$; see for instance \cite{Klevtsov, KMMW, KlevtsovZvonkine}. 
\end{remark}

\smallskip
\noindent
\emph{Acknowledgement}.  Our work was partially  supported by the  German Research Foundation SFB-TRR 358/1 2023 — 491392403, Initiative d’excellence
(Idex) program and the Institute for Advanced Study Fellowship of the University of
Strasbourg, and the ANR-20-CE40-0017 grant.

\end{document}